\documentclass[twocolumn,longbibliography,nofootinbib]{revtex4-1}
\usepackage{amsmath} 
\usepackage{amsfonts} 
\usepackage{amssymb} 
\usepackage{graphicx}
\usepackage{psfrag}
\usepackage{enumitem}
\usepackage[normalem]{ulem}

\newcommand{\id}{\mathrm{id}}

\newcommand{\suchthat}{\,|\,}
\newcommand{\set}[2]{ \{#1\suchthat #2\} }

\newcommand{\sset}[1]{ \{#1\} }
\newcommand{\tr}{\mathrm{Tr}}

\newcommand{\sub}{\subseteq}

\newcommand{\compt}{\diamond}

\newcommand{\stoc}[2]{(#1,#2)}
\newcommand{\stocg}[2]{(#1,#2)_\mathbf g}
\newcommand{\synd}[1]{\overline{#1}}
\newcommand{\chop}[1]{\mathrm {synd}(#1)}
\newcommand{\chan}[1]{{#1}}
\newcommand{\schan}[1]{\mathcal{#1}}
\newcommand{\cchan}[1]{\mathbf {#1}}
\newcommand{\appr}[1]{\triangleleft_{#1}}
\newcommand{\dist}[2]{\mathrm{d}({#1},{#2})}
\newcommand{\fail}[1]{\mathrm {fail}\,(#1)}
\newcommand{\supp}[1]{\mathrm {supp}\,#1}

\newcommand{\qed}{\hspace*{\fill}$\blacksquare$}
 
 \newtheorem{thm}{Theorem}
 \newtheorem{lem}[thm]{Lemma}
 
 \newtheorem{defn}[thm]{Definition}
 \newtheorem{prop}[thm]{Proposition}
 \newenvironment{proof}{\noindent \emph{Proof.}}{\qed}

\begin{document}

\title{Resilience to time-correlated noise in quantum computation}

\author{H\'ector Bomb\'in}
\affiliation{Yukawa Institute for Theoretical Physics, Kitashirakawa Oiwakecho, Sakyo-ku, 606-8502 Kyoto}
\altaffiliation{
Center for Gravitational Physics, Yukawa Institute for Theoretical Physics, Kyoto University
}

\begin{abstract}
Fault-tolerant quantum computation techniques rely on weakly correlated noise. 
Here I show that it is enough to assume weak spatial correlations: time correlations can take any form. 
In particular, single-shot error-correction techniques exhibit a noise threshold for quantum memories under spatially local stochastic noise.
\end{abstract}


\maketitle

\section{Introduction}

Building a quantum computer amounts to attain exquisite control of a large quantum system.
In particular, it requires quantum systems that (i) are isolated from the rest of the universe and (ii) can be arbitrarily manipulated, all with any desired accuracy.
\emph{A priori}, quantum systems naturally displaying this unlikely combination of characteristics might exist~\cite{bombin:2013:self}.
Unfortunately, for the time being no such systems are known~\cite{brown:2014:quantum}.
Instead, one has to deal with decoherent quantum systems and innacurate control, \emph{i.e.} with noise.

The great success of the pioneering works in fault-tolerant quantum computation~\cite{shor:1996:ftqc,aharonov:1997:ftqc,kitaev:1997:quantum, knill:1998:resilient} was to show that a small amount of noise does not represent an insurmountable difficulty:
As long as noise is weak and weakly-correlated in space and time, quantum computations of arbitrary accuracy and complexity can be carried out.
The purpose of this paper is to drop one of the conditions imposed on noise, by allowing arbitrary correlations in time.
The motivation is to extend the class of physical systems potentially useful for quantum computation.

\subsection{Fabrication faults}

Suppose that a quantum computer is built out of many interconnected pieces or `nodes', each node holding a low-dimensional and noisy quantum system, each connection allowing the controlled (but noisy) interaction of the corresponding quantum systems.
The system is designed to implement some fault-tolerant computation scheme, so that when all the nodes (and connections) work as expected, a certain degree of noise is tolerated.

This work originally grew out of an interest in the following question: what will happen if a fraction of the nodes fail permanently? Are there fault-tolerant schemes resilient to such `fabrication faults'?
The answer could, in principle, depend on important details, such as whether the fabrication faults are known or not, or how flexible the operation of the system is. 
However, the study of general time-correlated noise automatically encompasses all such possible situations: 
If a system can deal with such correlations, then it can also deal with unknown fabrication faults, with no need for alternative operations.

In fact, one can go far beyond fixed fabrication faults, and imagine faults that fluctuate slowly over time and introduce significant time-correlated noise. 
There is no guarantee that a conventional approach to fault-tolerant quantum computation could function within such conditions. 
Thus, demonstrating fault-tolerant techniques that are compatible with time-correlated noise opens the door to the use of physical systems exhibiting such problematic behavior. 

\subsection{Topological codes}

Topological quantum error correction~\cite{dennis:2002:tqm} is a popular approach to fault tolerance that emphasizes the locality of interactions in the above setting: nodes form a lattice of a certain spatial dimensionality, with connections between nearby nodes only.
For practical reasons, it is desirable that the spatial dimensionality of the lattice is as low as possible, which puts the focus on 2D topological codes such as toric codes~\cite{kitaev:1997:quantum} and color codes~\cite{bombin:2006:2DCC}.

It is not known whether 2D stabilizer codes are generally resilient to fabrication faults, even when the locations of the faults are known and it is possible to change the local operation of the physical qubits, but the problem is already under study~\cite{nagayama:2016:surface}. 
On the other hand, it will be shown below that there exist 3D topological schemes, namely 3D gauge color codes~\cite{bombin:2015:gauge}, capable of dealing not only with (dilute enough) unknown fabrications faults, but also with noise exhibiting arbitrary correlations in time.
It is possible to argue, see appendix~\ref{app:difficulties}, that 2D codes are not resilient to noise with such correlations, even if they have a Markovian origin.

\subsection{Single-shot error correction}

In a fault-tolerant quantum computer information is protected by means of redundancy: instead of using all available degrees of freedom for storing information, some (or most) of them are used to absorb the damage inflicted by noise.
This cannot work indefinitely, because the effect of noise piles up over time.
Therefore a process known as error correction must be performed repeatedly to flush away the errors inflicted to the system, making room to subsequently absorb more errors.

Error correction consists of two stages.
The first stage identifies the errors most likely suffered by the system, and the second one corrects them.
It is critical that errors are identified correctly, because otherwise the correction step might end up being counterproductive.
In particular, in the first stage a number of operators are measured, providing an `error syndrome'. 
In the conventional approach, to guarantee that the error syndrome is not too noisy these measurements are repeated a number of times.
This is, in particular, the approach taken for 2D topological codes.

In the presence of time-correlated noise, repeating measurements ceases to be viable and alternative approaches are needed.
Single-shot error correction is such an alternative~\cite{bombin:2015:single-shot}.
It relies on an unusual feature of some error-correcting schemes, for which errors occurring at the measuring stage that are localized in space give rise to errors induced by the correction stage that are similarly localized in space.
This means that error correction can be performed directly in a single step (the `single shot'), with no need to repeat measurements to avoid errors. 
In fact, in the case of 3D gauge color codes~\cite{bombin:2015:gauge} much more is true, since single-shot error correction enables the performance of arbitrary elementary computational operations in constant time, irrespective of the system size. 
This is in stark contrast to most fault-tolerant settings, where some operations will incur in a time-overhead proportional to the amount of errors that the system can correct. 

The purpose of this work is to show that single-shot error correction is compatible with time-correlated noise. 
In particular, this will be demonstrated for the error-correction schemes considered in~\cite{bombin:2015:single-shot}.
More general approaches to singles-shot error correction are conceivable, as discussed in appendix~\ref{app:extrinsic}.

\subsection{Spatially local stochastic noise}

For concreteness, consider a quantum computation performed on a number of qubits.
A computation is divided into time steps.
Ideally, at each step a number of separate processes happen simultaneously, each involving only a bounded number of qubits.
Each of this processes is called a location.
In modeling fault-tolerant quantum computation, the resulting ideal process of computation is deformed to introduce noise.
Although more sophisticated approaches exist~\cite{terhal:2005:fault,aliferis:2006:quantum}, here noise will be modeled stochastically~\cite{aharonov:1997:ftqc,knill:1998:resilient}, as explained next.
This approach is ideally suited to explore the effects of strongly time-correlated noise.

In the stochastic noise model it is assumed that a given set of locations fails with some probability, while the rest behave ideally.
In particular, at a given time step the qubits that belong to faulty locations undergo a possibly unknown process, possibly involving the environment.

An essential assumption in fault-tolerant quantum computation is that noise is both weak and local: weakly correlated in space and time.
In the stochastic model the locality of noise is implemented as follows: The probability that any given set of $n$ locations fails, irrespective of what happens at other locations, is bounded by $p^n$ for some error rate $p$.

If nothing is known about time correlations, the only meaningful constraints on the distribution of faulty locations have to involve simultaneous locations, \emph{i.e.} belonging to the same time step.
Here a spatially local stochastic noise model will be adopted: The probability that any given set of $n$ simultaneous locations fails is bounded by $p^n$ for some error rate $p$.

Fabrication faults can also be naturally modeled in the stochastic noise model. 
A given set of fabrication faults can be related to a set of qubits: For every time step, all locations involving any of those qubits fail.
If the fabrication faults are fixed, this does not yield a spatially local distribution of faults.
However, the fabrication faults will typically be randomly distributed.
In particular, if the fabrication faults follow a certain spatial distribution (such as a local one), the faulty locations will follow a closely related spatial distribution.
This is how the study of fabrication faults fits into the study of stochastic noise with arbitrary time correlations. 
It suffices to focus on the latter, which is also much more general.

\subsection{Results}

The focus of this paper is to show that the single-shot error-correction techniques introduced in~\cite{bombin:2015:single-shot} are compatible with spatially local stochastic noise.
In particular, the corresponding quantum memories exhibit a noise threshold under local stochastic noise, in the following sense.
In a quantum memory error correction is repeated time after time~\cite{dennis:2002:tqm}, with no actual computation.
Any fault-tolerant scheme always involves a family of schemes requiring a varying number of qubits: if more accuracy is desired, a larger set of qubits is needed.
Quantum memories have a threshold $p_0$ for a given family of schemes and a noise model with parameter $p$ when, for any $p<p_0$, quantum information can be stored with arbitrary accuracy for arbitrarily long times, by choosing the right scheme in the family.

Even though only quantum memories will be discussed, there seems to be no difficulty in extending the results to general quantum computations, at least in the case of fault-tolerant schemes that use 3D gauge color codes~\cite{bombin:2015:gauge}.
The reason is that in the case of 3D gauge color codes the computation can be performed by alternating operations involving a single time step and single-shot error correction~\cite{bombin:2015:single-shot}.
Moreover, the computations can be performed in a network of qubits with 3 spatial dimensions~\cite{bombin:2016:dimensional}, so that, in principle, such a scheme could be realized with spatially local interactions.
For a numerical study of error thresholds in gauge color codes, see~\cite{brown:2016:fault}.

\section{Error correction}\label{sec:ec}

This section introduces some basic notions of quantum error correction and fault-tolerant quantum computation.

\subsection{Operator quantum error correction}\label{sec:oqec}

The purpose of quantum error-correction techniques~\cite{lidar:2013:quantum} is to send quantum information as faithfully as possible through a noisy quantum channel $E$.
This is achieved by choosing an encoding channel $C$ and a decoding channel $D$ such that the composed channel
\begin{equation}
DEC
\end{equation}
has as little noise as possible.
There is a price to be paid: the composed channel operates on a system with smaller dimensionality.

In operator quantum error correction~\cite{kribs:2006:operator} the Hilbert space of the quantum system where the channel $E$ operates is structured as follows:
\begin{equation}\label{eq:H}
H=H_0\oplus H',\qquad 
H_0 =  H_{\mathrm{ logical}}\otimes H_{\mathrm {gauge}}.
\end{equation}
The subspace $H_0$ is the code subspace.
It is divided into two subsystems, logical and gauge.
Quantum information is encoded in the logical subsystem, whereas the gauge subsystem is not to be protected and will typically be in a random state.
Although gauge degrees of freedom are not particularly interesting at the level of purely ideal error correction, they are surprisingly useful in the context of fault-tolerant quantum computation, both within topological techniques~\cite{bombin:2010:subsystem,bombin:2015:gauge,bombin:2015:single-shot,bombin:2016:dimensional,bravyi:2013:subsystem,bravyi:2015:doubled,jochym:2015:stacked} and elsewhere~\cite{bacon:2006:quantum,paetznick:2013:universal,anderson:2014:fault,bacon:2015:sparse}.

Given the above decomposition of $H$, the encoding channel $\chan C$ maps states on $H_{\mathrm {logical}}$ to states on $H$, and $D$ does the converse.
The encoding channel $C$ maps a logical state $\rho$ to the encoded state
\begin{equation}\label{eq:C}
C(\rho):=\rho\otimes \tau
\end{equation} 
on the code subspace $H_0$, with $\tau$ some (fixed) state of the gauge degrees of freedom.
The decoding channel $D$ takes the form
\begin{equation}\label{eq:D}
D(\rho):=\tr_\mathrm{gauge} \,R(\rho),
\end{equation} 
where the error-recovery channel $R$, which operates on the system $H$ and yields states with support on $H_0$, aims to undo the errors introduced by the channel $E$.

The linearity of error correction is rather useful.
For a fixed logical subsystem~\eqref{eq:H} and a fixed decoding channel~\eqref{eq:D}, there exists a subspace $V$ of operators on $H$ that characterizes correctable channels $E$, \emph{i.e.} those for which $DEC$ is the identity channel for any encoding $C$ as in~\eqref{eq:C}.
In particular, a channel $E$ is correctable if and only if its Kraus operators belong to $V$.
Another useful fact is that not all encodings need to be checked: If $DEC$ is the identity channel for the encoding $C$ with a completely random gauge state, then the same is true for any other encoding~\cite{kribs:2006:operator}.

\subsection{Quantum memories}

The purpose of fault-tolerant quantum computation techniques~\cite{lidar:2013:quantum} is to perform quantum computations of arbitrary size and precision by using faulty devices.
At the core of these techniques is error correction.
In fact, during the whole computation the logical information remains encoded in a subsystem as in~\eqref{eq:H}, in the following sense: 
If the ideal decoding operation $D$ were to be applied at any step of the computation, the result would be very close to the intended logical state at that computational step.
To stay within such a regime the errors that accumulate during the computation have to be flushed away regularly.
This can be achieved with a fault-tolerant recovery operation, a noisy analogue of the ideal recovery operation $R$.

The purpose of a quantum memory is to preserve some (logical) quantum state, rather than to compute with it.
A possible approach to this problem is to use a subset of the techniques of fault-tolerant quantum computation.
In particular, it suffices to apply as regularly as needed the fault-tolerant recovery operation.
For any given error-correcting code this only provides a limited memory lifetime, but typically codes come in families and the lifetime can be increased at the expense of using a larger system $H$.

The above encoding and decoding channels $C$, $D$ come handy in analyzing fault-tolerant protocols, and this is also true for quantum memories.
Assume that the effect of the passage of time $t$ in the quantum memory can be modeled as a quantum channel $E_t$.
The usual approach to analyze the quality of the memory is to consider the composed channel
\begin{equation}
DE_tC.
\end{equation}
In other words, one assumes perfect encoding and decoding in order to to analyze the effective noise on the logical degrees of freedom.

\section{Stochastic noise}\label{sec:stochastic}\label{sec:stochastic_noise}

This section introduces a formalism that allows to work comfortably with time-correlated stochastic noise.

\subsection{Preamble}

A quantum circuit on a set of qubits can be divided into time steps, each of which is itself further composed of locations: each location involves only a few qubits, which are either initialized, transformed or measured, with no two simultaneous locations involving the same qubits.
The circuit represents a quantum channel, which is obtained by composing the channels represented by each time step, which in turn are obtained by tensoring the channels represented by the corresponding locations.

In a stochastic noise model any given quantum circuit is modeled as a quantum channel of the form
\begin{equation}\label{eq:channel}
\sum_i{p_i}{ \chan E_i},
\end{equation}
where each $E_i$ is a channel corresponding to a certain set of faults in the circuit, or \emph{fault path}, and $p_i$ is the probability assigned to that particular path.
Each channel $E_i$ is a composition of channels $E_{i,t}$ representing the $t$-th time step.
$E_{i,t}$ behaves as expected on the locations that are not faulty, and the rest of qubits undergo some other process that might be unknown and path dependent\footnote{
The process on the qubits at faulty locations can incorporate an environment, which should vary in size over time: the input and output systems of $E_i$ are the same as for the noiseless circuit.
}.
The probabilities $p_i$ might also not be precisely known either: instead, they are assumed to satisfy some properties.
Because of these uncertainties, one usually deals with a set of channels of the form~\eqref{eq:channel}, rather than a single channel, and has to extract properties that are common to all the channels in the set.

When studying the behavior of a large circuit it is useful to be able to deal with its parts separately.
This section develops a formalism to do precisely this, in particular in the case of arbitrary and unknown time correlations on the fault probabilities.

\subsection{Stochastic channels}

A \emph{stochastic channel}
\begin{equation}\label{eq:sc}
\schan E=\stoc{p_i}{ \chan E_i}
\end{equation}
consists of a probability distribution $p_i$ over a finite collection of channels $\chan E_i$, 
all of which share a given input quantum system and a given output quantum system.
Given such a distribution there exists a quantum channel~\eqref{eq:channel} obtained by applying $\chan E_i$ with probability $p_i$.
But given a quantum channel, in general there will be different distributions that can produce it in this way.
Thus it is important not to confuse the stochastic channel $\stoc{p_i}{\chan E_i}$ with the corresponding channel~\eqref{eq:channel}.
Besides, the stochastic channel can be identified with the expression~\eqref{eq:channel} when the latter is regarded as a \emph{formal} sum, and this is sometimes useful.

It is convenient to identify any channel $\chan E$ with the stochastic channel where $\chan E$ is assigned probability $1$.
When considering distributions $p_{ij}$ over several variables (or similar objects), traces will be denoted with a compact notation:
\begin{equation}
p_{i(j)}:=\sum_{j} p_{ij}.
\end{equation}

A \emph{stochastic class} $\cchan A$ is a set of stochastic channels with shared input and output quantum systems. 
Their interest will become clear in the next section, which discusses noise models as maps taking quantum circuits to stochastic classes.
It is convenient to identify any stochastic channel $\schan E$ with the singleton class $\sset{\schan E}$.

\subsection{Correlated composition}

Suppose that a given circuit is split in time slices $t=1,\dots,n$, and for each time $t$ there is a number of fault configurations that can happen.
For a given fault path, the circuit maps to a channel
\begin{equation}
{\chan A^{(n)}_{i_n}}\cdots{\chan A^{(1)}_{i_1}},
\end{equation}
where $A^{(t)}_i$ is the channel characterizing the time slice $t$ given that the fault $i$ happens.
The effect of the circuit is characterized by some stochastic class $\cchan A$, with elements of the form
\begin{equation}\label{eq:comp_channel}
\schan E= \stoc{p_{i_1\cdots i_n}}{{{\chan A^{(n)}_{i_n}}\cdots{\chan A^{(1)}_{i_1}}}}.
\end{equation}
Eventually the goal is to deal with situations where the only available information concerns the probabilities for \emph{simultaneous} faults to happen.
In the case of the above stochastic class $\cchan A$ this information can be encoded in stochastic classes $\cchan A^{(t)}$, one per time slice $t$.
Namely
the elements of $\cchan A^{(t)}$ are stochastic channels of the form
\begin{equation}\label{eq:comp_channel2}
\stoc {p_{(i_1\cdots i_{t-1})i_t(i_{t+1}\cdots i_n)} }{A^{(t)}_{i_t}}
\end{equation}
for any given $\schan E$ as in~\eqref{eq:comp_channel}.

If only the stochastic classes $\cchan A^{(t)}$ are known, in general $\cchan A$ cannot be reconstructed. 
This lack of knowledge is desirable if we are only interested in properties arising from the $\cchan A^{(t)}$.
In such a situation, it is convenient to consider the most general stochastic class $\cchan A'$ that would yield such $\cchan A^{(t)}$, \emph{i.e.} the unique stochastic class $\cchan A'$ such that $\cchan A\subseteq \cchan A'$ for any $\cchan A$ as above.
This stochastic class $\cchan A'$ is denoted
\begin{equation}\label{eq:correlated_composition}
\cchan A^{(n)}\compt\cdots\compt\cchan A^{(1)}.
\end{equation}
In other words, given the $\cchan A^{(t)}$, the expression \eqref{eq:correlated_composition} is defined as the set of stochastic channels $\schan E$ that can be put in the form~\eqref{eq:comp_channel} with~\eqref{eq:comp_channel2} an element of $\cchan A^{(t)}$ for each $t$.

A priori, this definition introduces an $n$-ary operation on the $\cchan A^{(t)}$ with no further structure.
As proven in appendix \ref{app:associativity}, the corresponding binary operation $\compt$ is associative and when iterated yields the $n$-ary operation.
The binary operation $\compt$ will be called correlated composition.

As a contrast, consider the case of noise completely uncorrelated in time.
Instead of~\eqref{eq:correlated_composition}, one ought to consider 
\begin{equation}\label{eq:uncorrelated_composition}
\cchan A^{(n)}\circ\cdots\circ\cchan A^{(1)},
\end{equation}
defined as the set of stochastic channels $\schan E$ that can be put in the form
\begin{equation}\label{eq:comp_channel_unc}
\schan E= \stoc{p^{(1)}_{i_1}\cdots p^{(n)}_{i_n}}{{{\chan A^{(n)}_{i_n}}\cdots{\chan A^{(1)}_{i_1}}}}, \qquad \stoc{p^{(t)}_{i}}{A^{(t)}_i}\in\cchan A^{(t)}.
\end{equation}
Notice that using stochastic channels is an overkill in this case:
If stochastic channels are regarded as ordinary channels, the composition~\eqref{eq:uncorrelated_composition} is just the set of all channels obtained by composing channels in the sets $\cchan A^{(t)}$.

\subsection{Spatially local noise}

In section~\ref{sec:ec} the quantum system $H$ used to encode quantum information had an entirely abstract nature.
However, in most cases $H$ will be composed of a number of physical subsystems, typically qubits.
In fact, that error correction is possible at all depends on the noise having some sort of structure, and the most common assumption is that noise is  spatially local: errors affecting a large number of subsystems are unlikely.
Let 
$\supp {E}$
denote the set of qubits in which the channel $E$ does not act trivially\footnote{
If there is an environment, only the qubits belonging to the system under control should be considered.
}.
In terms of stochastic classes, spatially local noise can be modeled as follows.
\begin{defn}
Given $\lambda\geq 0$, $\cchan \Lambda_{\lambda}$ is the set of stochastic channels $\stoc{p_i}{E_i}$ such that for any set of qubits $R$ 
\begin{equation}\label{eq:spacelocal}
\sum_{i\suchthat R\subseteq \supp{E_i}} p_i \leq \lambda^{|R|}
\end{equation}
\end{defn}

How much can time correlations disrupt the locality of noise? 
A quantitative answer is provided in appendix~\ref{app:local_composition}:
\begin{equation}\label{eq:Lambda1}
\cchan \Lambda_\lambda\compt\cchan \Lambda_{\lambda'}\subseteq \cchan \Lambda_{2\max(\lambda,\lambda')^{1/2}}.
\end{equation}
As a comparison, for uncorrelated spatially local noise
\begin{equation}\label{eq:Lambda2}
\cchan \Lambda_\lambda\circ\cchan \Lambda_{\lambda'}\subseteq \cchan \Lambda_{\lambda+\lambda'}.
\end{equation}
Although time-correlated spatially local noise behaves much worse than time-uncorrelated spatially local noise, in both cases locality is preserved when composing several time slices together.

 \subsection{Approximate behavior}

Since stochastic channels are probability distributions over channels, the distance between them can be defined as the statistical distance.
Namely, it is always possible to write two stochastic channels $\schan A$ and $\schan B$ in the form
\begin{equation}
\schan A= \stoc{p_i}{E_i}, \qquad \schan B= \stoc{p'_i}{E_i},
\end{equation}
and their distance is
\begin{equation}
\dist{\schan A}{\schan B}:=\frac 1 2\sum_i |p_i-p'_i|.
\end{equation}
In some cases this distance might be huge compared to, say, the diamond distance of the corresponding channels.
However, since the aim is to deal with noise of unknown form, there is hardly any gain in considering a more sophisticated distance that accounts for the proximity of channels.
Moreover, as in any approach to error correction that models noise stochastically, eventually the goal is to bound the probability that a logical error might occur, and  the statistical distance does the job.

\begin{defn}\label{def:appr}
A class of stochastic channels $\cchan A$ is $\epsilon$-approximated by another class $\cchan B$, denoted 
\begin{equation}
\cchan A \appr\epsilon \cchan B,
\end{equation}
when for every $\schan A\in\cchan A$ there exists $\schan B\in\cchan B$ with
\begin{equation}
\dist{\schan A}{\schan B} \leq \epsilon.
\end{equation}
\end{defn}

The following basic properties are proven in appendix~\ref{app:aprox}:
\begin{align}
\label{eq:subset_appr}
\cchan A\appr 0 \cchan B\,\,\iff\,\,&\cchan A\subseteq \cchan B,
\\
\label{eq:transit_appr}
\cchan A\appr\epsilon \cchan B\appr\delta\cchan C\,\,\implies\,\,& \cchan A\appr{\epsilon+\delta}\cchan C,
\\
\label{eq:compt_appr}
\cchan A\appr\epsilon \cchan C\,\,\wedge\,\,\cchan B\appr\delta \cchan D\,\,\implies\,\, &\cchan A\compt\cchan B\appr{\epsilon+\delta} \cchan C\compt\cchan D.
\end{align}
Notice that the wedge $\wedge$ represents a logical `and'.

\subsection{Error rate}

As in section~\ref{sec:ec}, consider a pair $C, D$ of encoding and decoding operations defined for some quantum system $H$.
The error rate of a stochastic channel $\schan E$ that operates on $H$ is defined to be
\begin{equation}
\fail{\schan E}:=\dist{D \circ\schan E\circ C}{\id_{\mathrm {logical}}}.
\end{equation}
To clarify the meaning of this figure of merit, let $\schan E=\stoc{p_i}{E_i}$.
It can be regarded as a conventional channel
\begin{equation}
E=\sum_i p_i E_i
\end{equation}
satisfying, when $p=\fail {\schan E}$,
\begin{equation}
DEC =(1-p) \id_\mathrm {logical} + p E',
\end{equation}
where $E'$ is some channel.
The error rate bounds the probability that logical information is damaged.

The error rate of a stochastic class $\cchan A$ is the worst possible error rate for an element of the class,
\begin{equation}
\fail {\cchan A}:=\sup_{\schan E\in\cchan A} \fail {\schan E}.
\end{equation}
Observe that
\begin{equation}\label{eq:fail_appr1}
D\circ\cchan A\circ C\appr\epsilon \id_\mathrm{logical}\iff\fail {\cchan A}\leq\epsilon,
\end{equation}
which using~\eqref{eq:transit_appr} gives
\begin{equation}\label{eq:fail_appr}
D\circ\cchan A\circ C\appr\epsilon D\circ \cchan B\circ C\implies\fail {\cchan A}\leq\epsilon+\fail {\cchan B}.
\end{equation}

\subsection{Modeling quantum memories}\label{sec:model}

The quantum memory model studied here consists of two processes that alternate in time: noise accumulation and single-shot error recovery.
Each process is described by a parametrized stochastic class: $\cchan L_\lambda$ represents the accumulation of noise in between error-correction steps, and $\cchan R_\eta$ represents the noisy error recovery steps.
The stochastic classes $\cchan L_\lambda$ and $\cchan R_\eta$ will be described in section~\ref{sec:memories}. 
For the time being it suffices to know that the parameters $\lambda$ and $\eta$ have a similar meaning: they indicate the amount of (spatially local) noise afflicting each process, with the value zero indicating a noiseless process.

The figure of merit is the error rate
\begin{equation}\label{eq:goal}
\fail{(\cchan R_\eta\compt\cchan L_\lambda)^{\compt n}}
\end{equation}
where $n$ is the number of iterations of the error accumulation and recovery steps.
The goal is to show that there exist thresholds $\lambda_0, \eta_0>0$ such that for any 
\begin{equation}\label{eq:threshold}
\lambda<\lambda_0, \qquad \eta<\eta_0,
\end{equation}
the error rate \eqref{eq:goal} can be efficiently made arbitrarily small for any value of $n$, so that the quantum memory lifetime can be as long as desired.

\subsection{Approach}\label{sec:approach}

In general, the strategy to bound the error rate \eqref{eq:goal} will involve an additional parametrized stochastic class $\cchan N_\tau$.
Its purpose is to characterize the residual noise after error recovery, with $\tau\geq 0$ the amount of noise.
In~\cite{bombin:2015:single-shot} the parameter $\tau$ is referred to as a `temperature': error recovery can be regarded as cooling down the system by removing the entropy introduced by errors.
Assume that for some $\lambda$, $\eta$, $\tau_i$ and $\delta_i$, 
\begin{align}
\cchan R_\eta\compt\cchan N_{\tau_2} \circ C &\appr{\delta_1} \cchan N_{\tau_1}\circ C,\label{eq:goal_RNN}\\
\cchan L_{\lambda}\compt \cchan N_{\tau_1} \circ C& \appr{\delta_2} \cchan N_{\tau_2}\circ C,\label{eq:goal_LNN}\\
\cchan L_{\lambda} \circ C& \appr{\delta_2} \cchan N_{\tau_2}\circ C,\label{eq:goal_LN}\\
\fail{\cchan N_{\tau_1}}&\leq\delta_3.\label{eq:goal_N}
\end{align}

The following series of relations are constructed by repeatedly applying~(\ref{eq:goal_RNN}-\ref{eq:goal_LN}) via~\eqref{eq:compt_appr}:
\begin{align}\label{eq:display_notation}
&(\cchan R_\eta\compt\cchan L_\lambda)^{\compt n}\circ C 
= 
(\cchan R_\eta\compt\cchan L_\lambda)^{\compt (n-1)}\compt \cchan R_\eta \compt \cchan L_{\lambda}\circ C 
\nonumber\\
&\quad\appr{\delta_2} 
(\cchan R_\eta\compt\cchan L_\lambda)^{\compt (n-1)}\compt \cchan R_\eta \compt \cchan N_{\tau_2}\circ C
\nonumber\\
&\quad\quad\appr{\delta_1} 
(\cchan R_\eta\compt\cchan L_\lambda)^{\compt (n-1)}\compt \cchan N_{\tau_1}\circ C
\nonumber\\
&\quad\quad\quad=
(\cchan R_\eta\compt\cchan L_\lambda)^{\compt (n-2)}\compt \cchan R_\eta \compt \cchan L_{\lambda}\compt \cchan N_{\tau_1}\circ C 
\nonumber\\
&\quad\quad\quad\quad\appr{\delta_2} 
(\cchan R_\eta\compt\cchan L_\lambda)^{\compt (n-2)}\compt \cchan R_\eta \compt \cchan N_{\tau_2}\circ C
\nonumber\\
&\quad\quad\quad\quad\quad
\dots
\nonumber\\
&\quad\quad\quad\quad\quad\quad\appr{\delta_1} 
\cchan N_{\tau_1}\circ C
\end{align}
According to~\eqref{eq:transit_appr} these relations imply
\begin{equation}\label{eq:ntimes}
(\cchan R_\eta\compt\cchan L_\lambda)^{\compt n}\circ C\appr{n(\delta_1+\delta_2)} \cchan N_{\tau_1}\circ C,
\end{equation}
and thus, using~(\ref{eq:fail_appr}, \ref{eq:goal_N})
\begin{equation}\label{eq:goal2}
\fail{(\cchan R_\eta\compt\cchan L_\lambda)^{\compt n}}\leq n(\delta_1+\delta_2)+\delta_3.
\end{equation}
The problem is to find conditions under which~(\ref{eq:goal_RNN}-\ref{eq:goal_N}) hold and $\delta_1$, $\delta_2$ and $\delta_3$ can be made arbitrarily small (in a resource-efficient manner).
Such conditions will be described in section~\ref{sec:single-shot}.

\section{Stabilizer formalism}

A complete account of stabilizer codes in the context of single-shot error correction can be found in~\cite{bombin:2015:single-shot}.
The purpose of this section is to explain some basic aspects that will be required in the main text.
 
\subsection{Stabilizer codes}\label{sec:stabilizer}

A stabilizer code is defined on a given number of physical qubits~\cite{poulin:2005:stabilizer}.
It can be characterized with two sets of Pauli operators:  `check operators' (or stabilizers) and `gauge operators'.
They fix the code subspace and its subsystem structure, discussed in section~\ref{sec:ec}, as follows.

The check operators all commute, and the Hilbert space is the sum of subspaces $H_i$ of the same dimension
\begin{equation}\label{eq:Hi}
H=\bigoplus_\sigma H_\sigma,
\end{equation}
each corresponding to a different set of eigenvalues of the check operators.
Each check operator has two eigenvalues, and only one is compatible with the code subspace.
The index $\sigma$ can be identified with the `error syndrome': the set of check operators with an eigenvalue incompatible with the code.
It is useful to regard $\sigma$ also as a binary vector: each entry corresponds to a check operator and has value $1$ when the check operator belongs to $\sigma$.
With this notation $H_0$ is the code subspace, and the binary addition $\sigma+\sigma'$ of any two syndromes is also a syndrome.

The code subspace is composed of a logical and a gauge subsystem, each equivalent to a number of qubits.
The decomposition is fixed by the gauge operators, which generate an algebra $\mathfrak G$: its elements map $H_0$ to itself, act trivially on the logical subsystem, and can arbitrarily transform the gauge subsystem~\cite{poulin:2005:stabilizer}.
In fact, a similar decomposition applies to all the subspaces $H_\sigma$.

Let a Pauli operation be a quantum channel $\rho\mapsto e\rho e^\dagger$, with $e$ some Pauli operator.
Ideal error correction, \emph{i.e.} the error recovery operation in~\eqref{eq:D}, takes the form
\begin{equation}\label{eq:ideal_corr}
R = \sum_\sigma \omega_\sigma\,\chan P_\sigma,
\end{equation}
where $P_\sigma$ denotes the map $\rho\mapsto h\,\rho\, h$, with $h$ the projector onto $H_\sigma$, and $\omega_\sigma$ is a Pauli operation such that
\begin{equation}
\omega_\sigma \chan P_\sigma = \chan P_0 \omega_\sigma.
\end{equation}
Ideal error correction is a two-step process: first the check operators are measured to recover an error syndrome $\sigma$, then a Pauli operator is applied to bring the system back to the code subspace.
The choice of correction operations $\omega_\sigma$ will be considered to be part of the description of a stabilizer code.
This fixes, in particular, the decoding channel $D$, which incorporates the recovery channel~\eqref{eq:ideal_corr}.

\subsection{Processing of noisy syndromes}\label{sec:noisy_syndromes}

Fault-tolerant error correction, as opposed to the ideal error correction described above, might involve a series of measurements of check operators or, more generally, of gauge operators~\cite{poulin:2005:stabilizer}.
In either case, since each measurement has a binary outcome the measurement results can be described with a binary vector $x$. 
This binary string needs to be processed in a classical computer to produce an error syndrome $\sigma=r(x)$, and then the correction operation $\omega_\sigma$ is applied.

Given a syndrome $\sigma$, there will be a set of measurement outcomes $x$ compatible with $\sigma$ in the absence of errors, and then it is natural to impose that $r(x)=\sigma$.
This constraint, in general, is not enough to fix the function $r$, as many outcomes will not be compatible with any syndrome in this way.
There is a constraint on $r$, however, that often holds and will be relevant in the next section.
If $x$ is compatible with $\sigma$ as stated, and if $y$ is an arbitrary binary vector of the same length as $x$, then
\begin{equation}\label{eq:r}
r(x+y) = \sigma+r(y),
\end{equation}
where addition is modulo 2.
In particular, this is true for all the single-shot error-correction strategies discussed in~\cite{bombin:2015:single-shot}.

A remark is in order. 
As far as preserving the logical state is concerned, there is no need to actually perform the correction operations $\omega_\sigma$ at the end of each round of error correction: the stabilizer formalism makes it possible to simply keep track of them. 
In fact, the optimal (but computationally more expensive) strategy is always to put together the measurements from different rounds in order to infer the errors, rather than interpreting each round separately, for the simple reason that time-local strategies are instances of global strategies.

\subsection{Ignoring the gauge operators}\label{sec:channels}

It is convenient to use a notation that ignores the gauge algebra $\mathfrak G$~\cite{bombin:2015:single-shot}. 
For reasons to be clarified below, stochastic channels of interest will mainly take the form
\begin{equation}\label{eq:generalG}
\stoc {p_i}{\sum_\sigma\chan E_{i\sigma}P_\sigma G_{i\sigma}},
\end{equation}
where $E_{i\sigma}$ are Pauli operations and $G_{i\sigma}$ are quantum channels with Kraus operators in $\mathfrak G$.
Any such stochastic channel will be denoted, omitting the $G_{i\sigma}$,
\begin{equation}\label{eq:general_}
\stocg {p_i}{\sum_\sigma\chan E_{i\sigma}P_\sigma}.
\end{equation}
This does not give rise to any ambiguity regarding the compositions $\compt$ and $\circ$ because ignoring the $G_{i\sigma}$ terms is consistent with channel composition~\cite{bombin:2015:single-shot}.
The failure rate is not ambiguous either.
In fact, for these stochastic channels it does not depend on the encoding $C$ used.
For fixed $p_i$ and $E_{i\sigma}$, the set of all channels of the form~\eqref{eq:generalG} will be denoted
\begin{equation}\label{eq:general_set}
\sset{\stocg {p_i}{\sum_\sigma\chan E_{i\sigma}P_\sigma}}.
\end{equation}

There are two main kinds of stochastic channels that will be of interest below.
\begin{defn}
The stochastic class $\cchan P$ is the set of all Pauli channels, \emph{i.e.} stochastic channels
\begin{equation}\label{eq:pauli}
\stocg {p_i}{\chan E_i},
\end{equation}
where $E_i$ are Pauli operations. 
\end{defn}

\begin{defn}
The stochastic class $\cchan Q$ is the set of stochastic channels of the form
\begin{equation}\label{eq:noisy_R}
\stocg{q_\sigma}{R_\sigma},
\end{equation}
where
\begin{equation}\label{eq:noisy_R2}
R_\sigma := \sum_{\sigma'}  {\omega_{\sigma+\sigma'}} P_{\sigma'}.
\end{equation}
\end{defn}
In section~\ref{sec:noise_model} class $\cchan P$ will be used to model noise occurring in between recovery steps, and class $\cchan Q$ will be used to model noisy error recovery.

\section{Quantum memories}\label{sec:memories}

This section describes the main result, the compatibility of single-shot error correction and spatially local stochastic noise.

\subsection{Noise model}\label{sec:noise_model}

As stated in section~\ref{sec:approach}, the quantum memory model is described by a pair of parametrized stochastic classes: $\cchan L_\lambda$ represents the accumulation of noise in between error-correction steps, and $\cchan R_\eta$ represents the noisy error recovery steps.
This section provides their specific form, which aims to be as simple as possible.

\begin{defn}\label{def:paulilocal}
Given $\lambda\geq 0$, $\cchan L_{\lambda}$ is the set of stochastic channels that can be put in the form
\begin{equation}
\stocg{p_i}{E_i}\in\cchan P,
\end{equation}
so that
\begin{equation}
\stoc{p_i}{E_i}\in\cchan \Lambda_\lambda
\end{equation}
\end{defn}
The noisy steps are modeled as local Pauli noise, a standard phenomenological model that allows easy computations within the context of stabilizer codes.
Because of linearity, see section~\ref{sec:oqec}, there is no loss of generality in considering local Pauli noise rather than general local noise, namely
\begin{equation}
\fail{\cchan A\compt \cchan \Lambda_\lambda\compt\cchan B}\leq\fail{\cchan A\compt \cchan L_\lambda\compt\cchan B},
\end{equation}
where the inequality accounts for the extra gauge operations allowed on the right-hand side\footnote
{Similarly, the environment can be ignored without loss of generality.
}.
It is assumed that the waiting times in between error-correction rounds are independent of the code size, so that the parameter $\lambda$ is fixed for the whole family of codes.

In order to model noisy measurements it is again enough to consider Pauli noise.
Single-shot error correction with stabilizer codes involves (i) a finite-depth circuit to perform measurements (generally of gauge operators), which yield a binary vector $x$, (ii) classical processing to obtain the error syndrome $\sigma=f(x)$, and (iii) application of the Pauli operation $\omega_\sigma$.
In general noise will affect all aspects of the process, but from a qualitative point of view it suffices to consider one kind of fault alone: classical errors afflicting the outcome of the measurements.
Indeed, local Pauli noise afflicting the physical qubits can be propagated forward without changing its local nature, and thus can be accounted for in the subsequent `layer' of local noise.

Notice that the classical measurement outcomes can be treated as qubits for which phase-flip errors are immaterial, so that they might only undergo bit-flip errors:
instead of the correct outcome $x$, the measurement yields a result $x+y$ for some bit string $y$.
Assuming that condition~\eqref{eq:r} holds, the error-correction step is represented by stochastic channels
\begin{equation}\label{eq:noisy_Ry}
\stocg{p_y}{R_{r(y)} }\in\cchan Q,
\end{equation}
where $p_y$ is the probability that the bit-flips $y$ happen.
The measurement of gauge generators is not explicit in~\eqref{eq:noisy_Ry} due to the notation introduced in~\ref{sec:channels}.
In an abuse of notation, identify any binary vector with the set of positions at which it has value $1$.
To model the locality of noise, the probabilities $p_y$ are subject to a condition with parameter $\eta$ analogous to~\eqref{eq:spacelocal}: 
\begin{defn}
Given $\eta\geq 0$, $\cchan R_{\eta}$ is the set of stochastic channels 
\begin{equation}
\stocg{p_y}{R_{r(y)}}\in \cchan Q
\end{equation}
such that for any binary vector $x$ representing a bit-flip configuration,
\begin{equation}\label{eq:local_bit-flip}
\sum_{y\supseteq x} p_y \leq \eta^{|x|}.
\end{equation}
\end{defn}

\subsection{Residual noise}\label{sec:residual_noise}

This section adapts the results on the residual noise given in~\cite{bombin:2015:single-shot} to the notation and tools used here, particularly the correlated composition and the distance relations.
It reviews both the conditions imposed to the residual noise
\begin{equation}
\cchan N_\tau\subseteq \cchan P,
\end{equation}
and specific examples satisfying them.

For a given family of codes of increasing size, the conditions on the residual noise are as follows:
(i) there are real functions $g_1(x)$ and $g_2(x,y)$ with limit $0$ at the origin and monotonically increasing (in the case of $g_2$, on either of its parameters when the other is fixed), and (ii) for each code there are real functions $f_1(x)$, $f_2(x,y)$ and $f_3(x)$, and for each $i=1, 2, 3$ there is a neighborhood of the origin within which $f_i$ goes to zero in the limit of large codes, and (iii) the following relations hold:
\begin{align}
\cchan L_{\lambda} & \appr{f_1(\lambda)} \cchan N_{g_1(\lambda)},\label{eq:condition_LN}\\
\cchan N_{\tau_1}\compt \cchan N_{\tau_2}& \appr{f_2(\tau_1,\tau_2)} \cchan N_{g_2(\tau_1,\tau_2)},\label{eq:condition_NNN}\\
\fail{\cchan N_{\tau}}&\leq f_3(\tau).\label{eq:condition_N}
\end{align}
Moreover, noise monotonically increases with $\tau$,
\begin{equation}\label{eq:condition_mono}
0\leq\tau\leq\tau'\implies \cchan N_\tau\subseteq \cchan N_{\tau'},
\end{equation}
and disappears for $\tau=0$,
\begin{equation}\label{eq:condition_zero}
\cchan N_0=\sset{\stocg{1}{\id_H}}.
\end{equation}

Take, for example, the relation \eqref{eq:condition_LN}.
Not only is it required that the residual noise $\cchan N_\tau$ can approximate local noise $\cchan L_\lambda$ with arbitrary precision in the limit of large codes (for $\lambda$ below a threshold).
In addition, any `temperature' $\tau>0$, no matter how low, can be achieved for some $\lambda>0$.

The above conditions do not address the issue of resources, \emph{i.e.} how fast the error bounds $f_i$ go to zero as a function of the number of qubits in the code.
In the case of the fault-tolerant schemes discussed below, a precision $f_i<\epsilon$ requires a polylogarithmic number of qubits in $1/\epsilon$, which is efficient.

\subsubsection{Local noise}\label{sec:local_noise}

For some codes it will suffice to use local noise as the model for residual noise.
In this case
\begin{equation}\label{eq:NL}
\cchan N_\tau= \cchan L_\tau,
\end{equation}
and thus a valid choice is, see appendix~\ref{app:local_composition},
\begin{alignat}{2}
f_1(\lambda)&=0, \qquad &g_1(\lambda)&=\lambda,\\
f_2(\tau_1,\tau_2)&=0, \qquad &g_2(\tau_1,\tau_2)&=2 \max(\tau_1,\tau_2)^{1/2}.
\end{alignat}
The only nontrivial condition comes from~\eqref{eq:condition_N}: the family of codes has to have a threshold for local noise.
Notice that in~\cite{bombin:2015:single-shot} local residual noise was handled in a slightly different manner, see appendix~\ref{app:effective}.

\subsubsection{Local syndromes}\label{sec:local_syndromes}

For some codes it is not residual noise that has the local structure~\eqref{eq:spacelocal}, but rather the residual distribution of error syndromes.
In such cases the residual noise model of interest is
\begin{equation}\label{eq:Nexc}
\cchan N_\tau = \cchan N_\tau^\mathrm{exc}:= \set{\stocg{q_\sigma}{\omega_\sigma}}{\forall\sigma,\,\sum_{\sigma'\supseteq \sigma} q_{\sigma'} \leq \tau^{|\sigma|}}.
\end{equation}
Recall that $\sigma$ represents a subset of check operators, and thus the definition strongly depends on the choice of check operators.
Since the elements of $\cchan N_\tau$ are composed of correction operations, in this case
\begin{alignat}{1}
f_3(\tau)&=0.
\end{alignat}

Appendix~\ref{app:local_syndromes} discusses conditions under which all the required properties are met.
As in~\cite{bombin:2015:single-shot}, such conditions are satisfied by topological stabilizer codes that are related to self-correcting Hamiltonian systems, in the following sense.
For any local stabilizer code one can always write down a local quantum Hamiltonian as a sum of energy penalties for check operators in the error syndrome.
The ground state of such a system is the subspace of encoded states, and the other eigenstates can be chosen to have well-defined syndromes.
The energy of such an eigenstate is the sum of the energy penalties for each of the check operators.
For some codes a confinement mechanism gives rise to self-correction~\cite{dennis:2002:tqm,alicki:2010:thermal}: in the thermodynamic limit, in thermal equilibrium and below a threshold temperature, the system becomes a perfect quantum memory without the need of any external intervention.
Definition~\eqref{eq:Nexc} is in fact motivated by thermal equilibrium, were states have a probability that decreases exponentially with the energy.

\subsection{Effective noise}\label{sec:single-shot}

This section gives an account of the last ingredients required to prove that single-shot error-correction strategies give rise to quantum memories of arbitrary lifetimes.

In order to obtain relations of the form \eqref{eq:goal_RNN} it is convenient to introduce the following function~\cite{bombin:2015:single-shot}.
It maps stochastic classes $\cchan R\subseteq \cchan Q$ to `effective' sets of Pauli channels $\mathrm{eff}(\cchan R)\subseteq\cchan P$.
\begin{defn}\label{def:eff}
For any $\cchan R\subseteq \cchan Q$, the stochastic class $\mathrm{eff}(\cchan R)$ contains the stochastic channels
\begin{equation}
\stocg{q_\sigma}{\omega_\sigma}
\end{equation}
such that $\cchan R$ contains an element of the form
\begin{equation}
\stocg{q_\sigma}{R_\sigma}.
\end{equation}
\end{defn}
Its usefulness stems from the following lemma.
\begin{lem}\label{lem:central}
For any stabilizer code, $\cchan R\subseteq \cchan Q$ and $\cchan E\subseteq\cchan P$
\begin{equation}\label{eq:central}
\cchan R\compt \cchan E\circ P_0 \appr{\epsilon} \mathrm{eff}(\cchan R)\circ P_0,
\end{equation}
where
\begin{equation}
\epsilon=\fail { \mathrm{eff}(\cchan R)\compt\cchan E}.
\end{equation}
\end{lem}
The proof is in appendix~\ref{app:proof}. 
The same result holds if the  correlated composition $\compt$ is substituted with the uncorrelated composition $\circ$.

Lemma~\ref{lem:central} provides the key condition for the quantum memory to have an arbitrarily long lifetime. 
Assume that for a given family of codes the residual noise $\cchan N_\tau$ satisfies all the conditions  of section~\ref{sec:residual_noise}, and that, in addition, the following is true:
(i) there is a real function $g_4(x)$ with limit $0$ at the origin, that is monotonically increasing; (ii) for each code there is a real function $f_4(x)$, and there is a neighborhood of the origin within which $f_4$ goes to zero in the limit of large codes; (iii) the following relation holds:
\begin{equation}
\mathrm{eff}(\cchan R_{\eta}) \appr{f_4(\eta)} \cchan N_{g_4(\eta)}.\label{eq:condition_RN}
\end{equation}
It is not difficult to check, see appendix~\ref{app:together}, that the relations~(\ref{eq:goal_RNN}-\ref{eq:goal_N}) hold setting
\begin{align}
\tau_1&=g_4(\eta),\nonumber\\
\tau_2&=g_2(g_1(\lambda),\tau_1),\nonumber\\
\delta_1&=2f_4(\eta)+f_2(\tau_1,\tau_2)+f_3(g_2(\tau_1,\tau_2)),\nonumber\\
\delta_2&=f_1(\lambda)+f_2(g_1(\lambda),\tau_1),\nonumber\\
\delta_3&=f_3(\tau_1).\label{eq:choice}
\end{align}
Moreover, by taking $\lambda$ and $\eta$ small enough the $\delta_i$ go to zero in the limit of large codes, so that the memory can preserve quantum information for arbitrarily long times according to the inequality~\eqref{eq:fail_appr}.

A single-shot error recovery strategy for a stabilizer code is composed of a collection of gauge operators to be measured (that can be measured with a finite-depth circuit), a syndrome recovery function $r$ and a choice of Pauli correction operations $\omega_\sigma$, all such that the corresponding parametrized stochastic class $\cchan R_\eta$ fulfills the above conditions.
Notice, in particular, that the residual noise after every fault-tolerant recovery step can be made arbitrarily small by reducing the noise in the recovery, a characteristic trait of single-shot error correction.

Two different kinds of single-shot error-correction strategies were introduced in~\cite{bombin:2015:single-shot}.
The first strategy is based on 3D gauge color codes, which are a class of 3D topological codes with unique characteristics~\cite{bombin:2015:gauge}.
Thanks to the gauge structure, for gauge color codes the residual noise is local, as in~\eqref{eq:NL}.
The second strategy is based on stabilizer codes exhibiting self-correcting properties, as discussed in section~\ref{sec:local_syndromes}.
For such codes the residual noise has the local syndrome structure $\cchan N^\mathrm{exc}_\tau$ of~\eqref{eq:Nexc}.
For more details, see appendix~\ref{app:effective}.

\section{Discussion}

The purpose of this work is to show that fault-tolerant quantum computation is still possible in the presence of noise with arbitrary time correlations.
In particular, quantum memories based on single-shot error-correcting codes exhibit a threshold for spatially local stochastic noise.
A natural next step is to obtain a full-fledged threshold theorem for universal computation.
Since quantum computation with 3D gauge color codes amounts to a series of transversal operations and single-shot error correction~\cite{bombin:2015:single-shot,bombin:2016:dimensional}, there is no reason to expect obstacles in this regard.

The role of spatial dimensionality in quantum error correction is intriguing.
It is known that fault-tolerant quantum computation is possible with local gates even for a single spatial dimension~\cite{aharonov:1997:ftqc}.
This suggests that spatial dimension plays no role, at least qualitatively.
However, it is not known if single-shot error correction can be performed with less than three spatial dimensions.
It might be the case that for two spatial dimensions and arbitrary time correlations, fault-tolerant quantum computation is not possible.
If this were true, it would set an interesting example of how spatial dimensionality can play a fundamental role in fault-tolerance.

One of the defining features of single-shot error correction is that it only requires a finite-depth quantum circuit, assisted with non-local classical computation.
In particular, no classical information needs to flow between different layers of single-shot error correction.
A stronger form of locality is achieved by removing the ability to process classical information nonlocally, thus reducing quantum error correction to purely local circuitry.
This seems to be incompatible with single-shot error correction: a single step of fully local error correction is unlikely to be able to produce arbitrarily low residual noise (in the limit where error correction is close to noiseless).
However, fault tolerance can still be achieved~\cite{ahn:2004:extending}: it suffices for residual noise to stay within certain bounds.
A natural question is whether such entirely local forms of error correction can deal with time-correlated noise (or at least with some particular forms of it).

Finally, it is interesting to ask, more generally, if there are other approaches to achieve resilience to time-correlated errors.

\begin{acknowledgments} 

I am grateful to Benjamin J. Brown and Michael J. Kastoryano for useful discussions, and to Aleksander Kubica for comments on an early version of the manuscript. 
I received support from the MINECO grant FIS2012-33152 and the CAM grant QUITEMAD+. 
This work was supported by the International Research Unit of Advanced Future Studies at Kyoto University.

\end{acknowledgments}

\appendix

\section{Difficulties with 2D codes}\label{app:difficulties}

This appendix is addressed to readers familiar with topological stablilizer codes.
It aims to briefly discuss why noise with arbitrary time correlations cannot be handled by 2D topological stabilizer codes.
These include the toric code~\cite{kitaev:1997:quantum} and 2D color codes~\cite{bombin:2006:2DCC}.
The important feature of these codes is that they have string-like logical operators for which \emph{a localized set of faults can hide the syndromes} at a given endpoint of a string-like error operator. 

The outline of the argument is as follows.
For any string-like logical operator, clearly there exists a history of faults that (i) implements the logical operator, producing a logical error, (ii) leaves no trace in the syndrome history, and (iii) only requires a finite number of faults per time step.
Such a fault-path does not satisfy the spatial locality constraint, but it suffices to consider an ensemble of a finite number of such fault-paths, each starting at a different time.
This limits the memory lifetime to be proportional to the length of the string, but in fact the argument can be improved by dropping condition (iii) above.
Instead, with a bit of care one can allow for a finite density of faults per time step, getting a finite memory lifetime.

This argument does not work for 3D gauge color codes because it is not possible to hide a topological charge with measurement faults localized in the vicinity of the charge.
Instead, the faults should be string-like and connect the charge to another one or to a boundary.
This is the confinement mechanism~\cite{bombin:2015:single-shot}.

Finally, it is worth noting that the above argument can be slightly modified so that the stochastic process governing the noise takes the form of a Markov chain.
As stated, certain histories of faults unavoidably damage the memory and only require a few faults per time step.
In the simplest case, the history amounts to moving a particle along a predetermined trajectory from one end of the system to another.
To satisfy the Markov property, it suffices to start with some probability such a history at every time step, that is, to randomly create particles at the starting end.
Indeed, (i) since the start of the history is random, it does not depend on what happened in the past, and (ii) once a history has started, every step  of it depends on what happened in the previous step (\emph{i.e.}, it depends on where the particle is).
Moreover, the trajectories of particles starting at different times do not interfere with each other.
Finally, to satisfy the spatial locality constraint it suffices to make the probability to start a history as small as necessary.

\section{Extrinsic single-shot error correction}\label{app:extrinsic}

As defined in~\cite{bombin:2015:single-shot}, an approach to noisy error correction can be considered to be single-shot if it only involves a quantum-local operation: a finite depth quantum circuit involving ancilla qubits and aided with global classical computation (implementing feedback on measurement results).
However, the single-shot error-correction techniques discussed in~\cite{bombin:2015:single-shot} follow a much more strict set of rules: the circuits are local with respect to geometry of the topological codes used, and the density of ancilla qubits (within that geometry) is bounded by a constant.
Such single-shot error-correction techniques can be regarded as intrinsic, as they strongly depend on special properties of the codes.
It might be worth considering a wider class of extrinsic techniques that do not rely so strongly on specific codes, but still fit within the wider definition.

In Knill's approach to error correction \cite{knill:2005:scalable} the measurement of check operators is performed via teleportation.
The aim of Knill's technique is to reduce the exposure to noise of those physical qubits where quantum information is encoded.
This requires the preparation of high quality states on the ancillas used for teleportation.
But the teleportation of check operator measurements could also be used differently.
For a local family of codes~\cite{kitaev:1997:quantum}, where check operators can be measured with a finite depth circuit, any number $n$ of noisy rounds of check operator measurements can be performed in constant time by performing $n$ teleportations in parallel.
This approach fits the wider definition of single-shot error correction, at the cost of using a large number of ancilla qubits: In the case of topological codes, it amounts to add an extra spatial dimension to the qubit lattice.
Notice that in the case of surface codes a similar alternative already exists in the form of 3D cluster state quantum computation~\cite{raussendorf:2007:deformation}.

\section{Stochastic channels}\label{app:stochastic_channels}

This appendix provides the results required in section~\ref{sec:stochastic_noise}.

\subsection{Associativity}\label{app:associativity}

\begin{defn}
Given classes $\cchan A$ and $\cchan B$, their correlated composition $\cchan A\compt\cchan B$ is the set of stochastic channels $\schan E$ of the form
\begin{equation}\label{eq:compt}
\begin{gathered}
\schan E= \stoc{p_{ij}}{{{\chan A_{i}}{\chan B_{j}}}},
\\
\stoc {p_{i(j)} }{A_{i}}\in\cchan A,\qquad
\stoc {p_{(i)j} }{B_{j}}\in\cchan B.
\end{gathered}
\end{equation}
\end{defn}

For short, when writing $\stoc{p_{ij}}{\chan A_i\circ\chan B_j}\in\cchan A\compt \cchan B$ it is understood that $p_{ij}$, $A_i$ and $B_j$ satisfy \eqref{eq:compt} (notice the explicit composition symbol $\circ$).

The following result shows that $\compt$ has the intended meaning.

\begin{prop}\label{prop:assoc}
The binary operation $\compt$ is associative.
In particular, 
\begin{equation}
\schan E\in\cchan A^{(1)}\compt\cdots\compt\cchan A^{(n)}
\end{equation}
if and only if $\schan E$ takes the form
\begin{equation}\label{eq:assoc}
\begin{gathered}
\schan E= \stoc{p_{i_1\cdots i_n}}{{{\chan A^{(1)}_{i_1}}\cdots{\chan A^{(n)}_{i_n}}}},\\
\stoc {p_{(i_1\cdots i_{k-1})i_k(i_{k+1}\cdots i_n)} }{A^{(k)}_{i_k}}\in\cchan A^{(k)},\quad k=1,\dots,n.
\end{gathered}
\end{equation}
\end{prop}

\begin{proof}
First we show by induction that
\begin{equation}\label{eq:preassoc}
\schan E\in ((\cdots(\cchan A^{(1)}\compt \cchan A^{(2)}) \compt\cdots)\compt \cchan A^{(n-1)})\compt\cchan A^{(n)}
\end{equation}
if and only if $\schan E$ takes the form \eqref{eq:assoc}.
The case $n=2$ is true by definition.
Given $j\geq 2$, assume that the statement holds for $n=j$ and let us show that it holds for $n=j+1$.
The if direction is trivial, and we omit it.
Given $\schan E$ as in \eqref{eq:preassoc}, by definition and by the inductive assumption it takes the form
\begin{equation}\label{}
\begin{gathered}
\schan E= \stoc{q_{lm}}{{{\chan B_{l}}{\chan A^{(j+1)}_{m}}}},
\\
\stoc {q_{(l)m} }{A^{(j+1)}_{m}}\in\cchan A^{(j+1)},
\\
\stoc {q_{l(m)} }{B_{l}}=
\stoc {r_{i_1\cdots i_{j}} }{{\chan A^{(1)}_{i_1}}\cdots{\chan A^{(j)}_{i_j}}},
\\
\stoc {r_{(i_1\cdots i_{k-1})i_k(i_{k+1}\cdots i_j)} }{A^{(k)}_{i_k}}\in\cchan A^{(k)},\quad k=1,\dots,j,
\end{gathered}
\end{equation}
where the $B_l$ are all different. 
The second equation implies that
\begin{equation}\label{}
\begin{gathered}
q_{l(m)} = \sum_{i_1\cdots i_j} \eta_{l,i_1\cdots i_j}\, r_{i_1\cdots i_j},
\\
\eta_{l,i_1\cdots i_j}:=
\begin{cases}
1,&\text{if ${{\chan A^{(1)}_{i_1}}\cdots{\chan A^{(j)}_{i_j}}}=B_l$}
\\
0,&\text{otherwise.}
\end{cases}
\end{gathered} 
\end{equation}
Taking
\begin{equation}
p_{i_1\dots i_{j}m} := \sum_l \frac{q_{lm}}{q_{l(m')}} \eta_{l,i_1\cdots i_j}\, r_{i_1\cdots i_j}
\end{equation}
gives, as desired, \eqref{eq:assoc}, noting that
\begin{align}
p_{i_1\dots i_{j}(m)} &= r_{i_1\dots i_j},
\nonumber\\
 p_{(i_1\dots i_{j})m} &=q_{(l)m},
\nonumber\\ 
\sum_{i_1\dots i_j} \eta_{l,i_1\dots i_j}\, p_{i_1\dots i_{j}m} &=q_{lm}.
\end{align}

To complete the proof it suffices to show that 
\begin{equation}
\schan E\in \cchan A^{(1)}\compt (\cchan A^{(2)} \compt\cchan A^{(3)})
\end{equation}
if and only if $\schan E$ takes the form \eqref{eq:assoc} with $n=3$; however, the proof is analogous to the $j=2$ step of the inductive argument and thus we omit it.
\end{proof}

A couple of comments on the notation are in order.
If $\cchan A$ and $\cchan C$ are sets of channels and $\cchan B$ is a stochastic class, then
\begin{equation}
\cchan A \compt \cchan B \compt \cchan C= \cchan A \circ \cchan B \circ \cchan C.
\end{equation}
As a special case of the above, if $A$ and $C$ are channels and $\schan B$ a stochastic channel, the composition
\begin{equation}
A \circ \schan B \circ C
\end{equation}
is a singleton stochastic class, but it will also be identified with its single element.

\subsection{Approximation}\label{app:aprox}

The statements \eqref{eq:subset_appr} and \eqref{eq:transit_appr} follow respectively (and trivially) from the following distance axioms:
\begin{equation}
\begin{gathered}
\dist{\schan A}{\schan B}=0\iff \schan A=\schan B,\\
\dist{\schan A}{\schan C}\leq \dist{\schan A}{\schan B}+\dist{\schan B}{\schan C}.
\end{gathered}
\end{equation}
As for \eqref{eq:compt_appr}, due to~\eqref{eq:transit_appr} it suffices to show that $\cchan A\compt\cchan B\appr{\epsilon} \cchan C\compt\cchan B$ and $\cchan C\compt\cchan B\appr{\delta} \cchan C\compt\cchan D$. We show the former, and the later is analogous.

Let $\schan E=\stoc{p_{ij}}{A_i \circ B_j}\in\cchan A\compt\cchan B$, with $p_i:=p_{i(j)}\neq 0$ without loss of generality.
By assumption there exists $\schan C\in\cchan C$, non-negative $\epsilon_i$, $\alpha_k$ and channels $C_k$ such that, in formal sum notation,
\begin{equation}
\schan C= \sum_i (p_{i}-\epsilon_i)A_i+\sum_k \alpha_k C_k, 
\qquad 
\epsilon_{(i)}=\alpha_{(k)}\leq \epsilon.
\end{equation}
Recall here that $\epsilon_{(i)}$ denotes $\sum_i \epsilon_i$.
Consider the formal sum
\begin{equation}
\begin{gathered}
\schan E':=\sum_{ij} p'_{ij} A_iB_j +\sum_{kj} q_{kj} C_kB_j, 
\\
p'_{ij}:=p_{ij}\left(1-\frac {\epsilon_i}{p_i}\right),
\qquad 
q_{kj}:= \sum_i \frac {p_{ij} \epsilon_i \alpha_k}{p_i \,\alpha_{(l)} }.
\end{gathered}
\end{equation}
It is easy to check that
\begin{equation}
p'_{(i)j}+q_{(k)j}=p_{(i)j}, 
\qquad
p'_{i(j)}=p_i-\epsilon_i,
\qquad
q_{k(j)}=\alpha_k,
\end{equation}
so that $\schan E'\in \schan C\compt \cchan B\subseteq \cchan C\compt\cchan B$.
Moreover, as required,
\begin{equation}
2\,\dist{\schan E}{\schan E'}\leq \sum_{ij} |p_{ij}-p'_{ij}|+q_{(kj)}=2\alpha_{(k)}\leq 2\epsilon.
\end{equation}

\section{Pauli operations}

This appendix compiles some properties and definitions related to Pauli operations required below.

\subsection{Error syndrome}

Given a stabilizer code each Pauli operation $\chan E$ has a well-defined error syndrome: there is a unique $\sigma$ such that
\begin{equation}\label{eq:defsyndpre}
\chan E \chan P_0 = \chan P_\sigma \chan E.
\end{equation}
For each Pauli operation $E$ and such $\sigma$, define
\begin{equation}\label{eq:defsynd}
\chop E:=\sigma,\qquad\synd{\chan E}:=\omega_\sigma.
\end{equation}
The following relations will be useful:
\begin{align}
\synd{\synd {\chan E} \,\synd {\chan D}}&=\synd{\chan E\chan D},\label{eq:prop_synd1}\\
\chop{\synd{ED}}&\subseteq \chop{\synd E}\cup\chop{\synd D}\label{eq:prop_synd2}.
\end{align}

\subsection{Groups}

Pauli operations form a group. For any pair of elements $\chan E, \chan D$,
\begin{equation}\label{eq:basic_pauli}
\chan E^2=\id, \qquad \chan E\chan D=\chan D\chan E.
\end{equation}
Given a stabilizer code, an important subgroup is the gauge group, which is generated by gauge operations.

A Pauli operation $\chan L$ is a logical operation if it cannot be detected through error correction, \emph{i.e.} if its syndrome is trivial:
\begin{equation}
\chop L = 0.
\end{equation}
Logical operations form a group.
The gauge group is the subgroup of logical operations that do not affect the logical subsystem.

\subsection{Failure rates}

If ideal error correction is applied after a single Pauli operation $\chan E$ hits an encoded state, the net result is a logical Pauli operation:
\begin{equation}
\sum_\sigma \omega_\sigma\,P_\sigma \,E\,P_0=\synd E\,E\,P_0.
\end{equation}
Here, $E$ is correctable if $\synd EE$ is in the gauge group. 
In particular
\begin{equation}
\fail{E}=\begin{cases}
0 & \text{if $E$ is correctable},\\
1 & \text{otherwise},
\end{cases}
\end{equation}
and for a Pauli channel,
\begin{equation}
\fail{\stocg{p_i}{E_i}}=\sum_i p_i \,\fail {E_i}.
\end{equation}

\begin{prop}\label{prop:fail}
For any Pauli operations $\chan E$, $\chan D$, 
\begin{align}
\fail{\synd{\chan D\chan E}E}&=\fail{\synd{\chan D}\chan E}\label{eq:fail_eq}.
\end{align}
\end{prop}

\begin{proof}
In this proof (\ref{eq:basic_pauli}, \ref{eq:prop_synd1}) will be used repeatedly and implicitly.
For any logical Pauli operation $\chan L$ and any correction operation $\omega$
\begin{equation}\label{eq:omegaL}
\fail{\omega \chan L}=\fail{\chan L},
\end{equation}
because
\begin{equation}
\synd{\omega \chan L} \omega \chan L = \omega\omega \chan L = \synd {\chan L} \chan L.
\end{equation}
Using \eqref{eq:omegaL} twice (the Pauli operation in parentheses is logical), gives \eqref{eq:fail_eq}:
\begin{multline}
\fail{\synd {\chan D \chan E}\chan E}=\fail{\synd{\chan D}(\synd{\chan D}\,\synd {\chan D \chan E}\chan E)}=\fail{\synd{\chan D}\,\synd {D E}E}=\\
=\fail{\synd{D E}( \synd{\chan D}\,\synd {D E}E)}=\fail{\synd{\chan D} E},
\end{multline}
\end{proof}

\section{Syndrome distributions}\label{sec:syndrome_classes}

The following construction was introduced in~\cite{bombin:2015:single-shot} to keep track of the syndrome distributions of Pauli channels in a given set.
Here it will be used differently, in particular by means of the proposition below.
\begin{defn}\label{def:eff}
For any $\cchan A\subseteq \cchan P$, the stochastic class $\synd{\cchan A}$ contains the stochastic channels
\begin{equation}
\stocg{p_i}{\synd E_i}
\end{equation}
such that $\cchan A$ contains an element of the form
\begin{equation}
\stocg{p_i}{E_i}.
\end{equation}
\end{defn}

\begin{prop}\label{prop:appr_fail}
Given a stabilizer code, for any $\cchan A\subseteq \cchan P$
\begin{equation}
\cchan A \appr{\fail{\cchan A}}\synd{\cchan A}.
\end{equation}
\end{prop}

\begin{proof}
For every $\schan A\in\cchan A$ there exist Pauli operations $E_i$, $E'_j$, with the $E_i$ correctable and the $E'_j$ uncorrectable, and channels $G_i$, $G'_j$ with Kraus operators in $\mathfrak G$, such that as a formal sum
\begin{equation}
\schan A = \sum_i p_i E_iG_i+\sum_j p'_j E'_jG'_j,\qquad p'_{(j)}= \fail{\schan A}.
\end{equation}
The following stochastic channel is an element of $\synd{\cchan A}$:
\begin{equation}
\sum_i p_i \synd{E_i}+\sum_j p'_j \synd{E'_j}.
\end{equation}
Since $\synd{E_i}=E_i(\synd {E_i}E_i)$, and the parentheses contains a gauge operation, the following is also an element of $\synd{\cchan A}$,
\begin{equation}
\schan A':=\sum_i p_i {E_i}G_i+\sum_j p'_j \synd{E'_j},
\end{equation}
and as required
\begin{equation}
\dist{\schan A}{\schan A'}\leq p'_{(j)}\leq\fail{\cchan A}.
\end{equation}
\end{proof}

\section{Locality and composition}\label{app:local_composition}

The composition of different stochastic classes with local features, namely $\cchan \Lambda_\lambda$, $\cchan L_\lambda$ and $\cchan N_\tau^\mathrm{exc}$, can be handled with the following abstract result.
\begin{prop}\label{prop:comp_local}
Let $F$ be a function with a domain set $\cchan X$ of channels closed under composition,  and a codomain that is the power set of some set $Y$, such that for any $A,B\in \cchan X$
\begin{equation}\label{eq:ruleF}
F(AB)\subseteq F(A)\cup F(B).
\end{equation}
Let $\cchan S_\epsilon$, $\epsilon\geq 0$, be the set of stochastic channels $\stoc{p_i}{E_i}$, $E_i\in \cchan X$, such that for any $R\subseteq Y$
\begin{equation}
\sum_{i | R\subseteq F(E_i)} p_i\leq \epsilon^{|R|}.
\end{equation}
Then for $\alpha, \beta\geq 0$,
\begin{align}
\cchan S_\alpha\compt\cchan S_{\beta}&\subseteq \cchan S_{2\max(\alpha,\beta)^{1/2}},\\
\cchan S_\alpha\circ\cchan S_{\beta}&\subseteq \cchan S_{\alpha+\beta}.\label{eq:circlocal}
\end{align}
\end{prop}

\begin{proof}
Given $(p_{ij},\chan A_i\circ \chan B_j)\in\cchan S_\alpha\compt\cchan S_\beta$, for any $S\subseteq Y$
\begin{align}\label{eq:comp_local1}
\sum_{i | S\subseteq F(\chan A_i)} p_{i(j)} \leq \alpha^{|S|},\qquad
\sum_{j | S\subseteq F(\chan B_j)} p_{(i)j} \leq \beta^{|S|}.
\end{align}
Let $\epsilon:=\max(\alpha,\beta)$.
The required result is
\begin{multline}\label{eq:sum_cases}
\sum_{i,j | R\subseteq F(A_i B_j)} p_{ij} \leq
\sum_{S\subseteq R}\, \sum_{i | S\subseteq F(A_i)}\,\sum_{j|R-S\subseteq F(B_j)} p_{ij} \leq
\\
\leq
\sum_{S\sub R}\min(\alpha^{|S|},{\beta}^ {|R-S|})
\leq
\sum_{S\sub R}\min({\epsilon}^{|S|},{\epsilon}^ {|R-S|})
\leq
\\
\leq
\sum_{S\sub R}{\epsilon}^{\max(|S|,|R-S|)})
\leq
2^{|R|}{\epsilon}^{|R|/2},
\end{multline}
where the first inequality makes use of~\eqref{eq:ruleF} and the second of~\eqref{eq:comp_local1}.

Given $\stoc{p_i}{A_i}\in\cchan S_\alpha$ and $\stoc{q_j}{ B_j}\in\cchan S_\beta$ their composition is an element of $\cchan S_{\alpha+\beta}$ because
\begin{multline}\label{eq:sum_cases2}
\sum_{i,j | R\subseteq F(A_i B_j)} p_iq_j \leq
\sum_{S\subseteq R}\, \sum_{i | S\subseteq F(A_i)}\,p_i\sum_{j|R-S\subseteq F(B_j)} q_{j} \leq
\\
\leq
\sum_{S\sub R}\alpha^{|S|}\beta^{|R-S|}
=(\alpha+\beta)^{|R|}.
\end{multline}
\end{proof}

If $\cchan X$ is the set of all channels and $F(A)=\supp A$, then 
\begin{equation}
\cchan S_\epsilon=\cchan \Lambda_\epsilon,
\end{equation}
which gives~(\ref{eq:Lambda1}, \ref{eq:Lambda2}).
If $\cchan X$ is the set of Pauli operations and $F(A)=\supp A$, then 
\begin{equation}
\cchan S_\lambda\subseteq \cchan L_\lambda,
\end{equation}
and conversely 
\begin{equation}
\stocg{p_i}{E_i}\in\cchan L_\lambda\implies \sset{\stocg{p_i}{E_i}}\cap\cchan S_\lambda\neq\emptyset. 
\end{equation}
It follows that
\begin{align}
\cchan L_\lambda\compt\cchan L_{\lambda'}&\subseteq \cchan L_{2\max(\lambda,\lambda')^{1/2}},\\
\cchan L_\lambda\circ\cchan L_{\lambda'}&\subseteq \cchan L_{\lambda+\lambda'}.
\end{align}
If $\cchan X=\cchan P$ and $F(A)=\chop A$, then 
\begin{equation}
\cchan N^\mathrm{exc}_\tau\subseteq \cchan S_\tau,
\end{equation}
and conversely for any $\cchan A\subseteq \cchan P$,
\begin{equation}
\cchan A\subseteq \cchan S_\tau\iff \synd{\cchan A} \subseteq\cchan N^\mathrm{exc}_\tau.
\end{equation}
Using proposition~\ref{prop:appr_fail}, it follows that
\begin{align}
\cchan N^\mathrm{exc}_\tau\compt\cchan N^\mathrm{exc}_{\tau'}&\appr{\fail{\cchan N^\mathrm{exc}_\tau\compt\cchan N^\mathrm{exc}_{\tau'}}} \cchan N^\mathrm{exc}_{2\max(\tau,\tau')^{1/2}},\label{eq:resNNN}\\
\cchan N^\mathrm{exc}_\tau\circ\cchan N^\mathrm{exc}_{\tau'}&\appr{\fail{\cchan N^\mathrm{exc}_\tau\circ\cchan N^\mathrm{exc}_{\tau'}}} \cchan N^\mathrm{exc}_{\tau+\tau'}.
\end{align}
The following proposition can be used to bound the above failure rates (the case of $\circ$ is analogous, and in any case can be found in~\cite{bombin:2015:single-shot}).

\begin{prop}\label{prop:B}
If there exists an integer $m$ and a set $B$ of sets of check operators such that (i) each element of $B$ contains $m$ check operators, and (ii)
for every syndromes $\sigma, \sigma'$,
\begin{equation}
\fail{\omega_\sigma\omega_{\sigma'}}=1\implies \exists \mu\in B\suchthat \mu\subseteq\sigma\cup\sigma',
\end{equation}
then
\begin{align}\label{eq:B}
\fail{\cchan N_\tau^\mathrm{exc}\compt\cchan N_{\tau'}^\mathrm{exc}}\leq|B|(2\max(\tau,\tau')^{1/2})^{m}.
\end{align}
\end{prop}

\begin{proof}
Let $\epsilon=2\max(\tau,\tau')^{1/2}$. 
Given 
\begin{equation}
\schan E=\stocg{q_{\sigma\sigma'}}{\omega_\sigma\circ\omega_{\sigma'}}\in\cchan N_\tau^\mathrm{exc}\compt\cchan N_{\tau'}^\mathrm{exc},
\end{equation}
its failure rate can be bounded using the same technique as in~\eqref{eq:sum_cases}:
\begin{multline}
\fail {\schan E}=\sum_{ij} q_{\sigma\sigma'} \,\fail{\omega_\sigma\omega_{\sigma'}} \leq 
\\
\leq \sum_{\mu\in B}\sum_{\mu'\subseteq \mu}\, \sum_{\sigma | \mu'\subseteq \sigma}\,\sum_{\sigma'|\mu-\mu'\subseteq \sigma'} q_{\sigma\sigma'} \leq
\sum_{\mu\in B} \epsilon^{|\mu|}
= |B| \epsilon^{m}.
\end{multline}
\end{proof}

\section{Local syndromes}\label{app:local_syndromes}

The residual noise model $\cchan N_\tau$ is well suited to self-correcting families of topological stabilizer codes.
In~\cite{bombin:2015:single-shot} such families were characterized with a series of conditions, from which further properties were then extracted.
Here only some of them need to be recalled.

To begin with,
\begin{equation}
\fail{\cchan L_\lambda}=f_1(\lambda),\qquad\synd{\cchan L_\lambda} \subseteq \cchan N_{g_1(\lambda)},
\end{equation}
for the following choices of $f_1$ and $g_1$, which are compatible with the requirement of section~\ref{sec:residual_noise},
\begin{equation}
f_1(\lambda) = n \left(\frac \lambda{\lambda_0}\right)^{bn^i} , \qquad g_1(\lambda) = v_1 \lambda^{1/v_2},
\end{equation}
where $n$ is the number of physical qubits, and the rest of symbols are positive constants that depend on the code family.
Due to proposition~\ref{prop:appr_fail}, condition \eqref{eq:condition_LN} is met.

Second, there exists a set $B$ of syndromes satisfying the assumption of proposition~\ref{prop:B} with
\begin{equation}
m=cn^j,\qquad |B|\leq k n \tau_0^{-cn^j},
\end{equation} 
where again all symbols but $n$ are positive constants that depend on the code family.
In particular, $kn$ is a bound on the number of check operators.
According to proposition~\ref{prop:B} and \eqref{eq:resNNN}, condition \eqref{eq:condition_NNN} is met for the choices
\begin{alignat}{1}
f_2(\tau_1,\tau_2)&=k n \left(\frac{2\max(\tau_1,\tau_2)^{1/2}}{\tau_0}\right)^{cn^j}, \label{eq:nexc_f2}\\
g_2(\tau_1,\tau_2)&=2 \max(\tau_1,\tau_2)^{1/2}.\label{eq:nexc_g2}
\end{alignat}

\section{Proof of lemma~\ref{lem:central}}\label{app:proof}

This lemma states that for $\cchan R\subseteq \cchan Q$ and $\cchan E\subseteq\cchan P$,
\begin{equation}\label{eq:centralp}
\cchan R\compt \cchan E\circ P_0 \appr{\fail { \mathrm{eff}(\cchan R)\compt\cchan E}} \mathrm{eff}(\cchan R)\circ P_0,
\end{equation}

\begin{proof}
Every element of $\cchan R\compt \cchan E\circ \chan P_0$ takes the form $\schan T\circ\chan P_0$ for some
\begin{equation}
\schan T = \stocg{p_{\sigma i}}{R_\sigma \circ E_i}\in\schan R\compt \schan E,
\end{equation}
where  $\schan R\in\cchan R$ and $\schan E\in\cchan E$.
Consider the set 
\begin{equation}\label{eq:F}
\cchan F_\schan T:=\sset{\stocg{p_{\sigma i}}{\synd {\omega_\sigma E_i}E_i}}.
\end{equation}
It satisfies
\begin{align}\label{}
\synd{\cchan F_\schan T}&= \sset{\stocg{p_{\sigma i}}{\synd{\synd{\omega_\sigma E_i}E_i}}}= 
\sset{\stocg{p_{\sigma(i)}}{\omega_\sigma}} =
\mathrm{eff} (\schan R),
\nonumber\\
\fail {\cchan F_\schan T}&=\sum_{\sigma i} p_{\sigma i} \,\fail {\synd {\omega_\sigma E_i}E_i}=
\nonumber
\\
&=\sum_{\sigma i} p_{\sigma i} \,\fail { \omega_\sigma E_i} 
\leq
\fail {\mathrm{eff}(\schan R)\compt\schan E},
\end{align}
where \eqref{eq:prop_synd1} was used in the computation of $\synd {\cchan F_\schan T}$ and~\eqref{eq:fail_eq} in that of $\fail{\cchan F_\schan T}$.
From (\ref{eq:defsyndpre}, \ref{eq:defsynd}) we get
\begin{equation}
\schan T\circ \chan P_0\in\cchan F_\schan T\circ \chan P_0.
\end{equation}
Let $\cchan F$ be union of all $\cchan F_\schan T$ of the form \eqref{eq:F} for some such $\schan T$.
Then
\begin{equation}\label{eq:central}
\begin{gathered}
\cchan R\compt \cchan E\circ \chan P_0\subseteq \cchan F \circ\chan P_0,
\\
\synd{\cchan F}\subseteq\mathrm{eff}(\cchan R),
\qquad
\fail {\cchan F}\leq\fail {\mathrm{eff}(\cchan R)\compt\cchan E},
\end{gathered}
\end{equation}
which via proposition \ref{prop:appr_fail} yields the desired result.
\end{proof}

The proof for the uncorrelated composition $\cchan R\circ \cchan E$ is entirely analogous; basically, it suffices to substitute $p_{\sigma i}$ with $q_\sigma p_i$.

\section{Putting it together}\label{app:together}

The purpose of this appendix is to show that assuming that $\cchan N_\tau$ and $\cchan R_\eta$ satisfy the relations~(\ref{eq:condition_LN}-\ref{eq:condition_zero}) and~\eqref{eq:condition_RN}, then the conditions~(\ref{eq:goal_RNN}-\ref{eq:goal_N}) are met for the choices~\eqref{eq:choice}.

The relation \eqref{eq:goal_RNN} follows, using~\eqref{eq:transit_appr}, from
\begin{multline}
\cchan R_\eta\compt\cchan N_{\tau_2} \circ C \subseteq 
\cchan R_\eta\compt\cchan N_{\tau_2} \circ P_0\circ C \appr{\epsilon}\\
\appr\epsilon \mathrm{eff}(\cchan R_\eta)\compt \cchan N_{\tau_2} \circ P_0\circ C \appr{f_4(\eta)}
\cchan N_{\tau_1}\circ C,
\end{multline}
where
\begin{multline}
\epsilon=\fail{\mathrm{eff}(\cchan R_\eta)\compt \cchan N_{\tau_2}}\leq 
f_4(\eta) +\fail{\cchan N_{\tau_1}\compt \cchan N_{\tau_2}}\leq
\\\leq
f_4(\eta) +f_2(\tau_1,\tau_2)+\fail{\cchan N_{g_2(\tau_1,\tau_2)}}\leq
\\\leq
f_4(\eta) +f_2(\tau_1,\tau_2)+f_3(g_2(\tau_1,\tau_2)).
\end{multline}
The relations (\ref{eq:goal_LNN}, \ref{eq:goal_LN}) follow from
\begin{multline}
\cchan L_{\lambda}\subseteq 
\cchan L_{\lambda}\compt \cchan N_{\tau_1} \appr{f_1(\lambda)} 
\cchan N_{g_1(\lambda)}\compt \cchan N_{\tau_1}\appr{f_2(g_1(\lambda),\tau_1)}
\cchan N_{\tau_2},
\end{multline}
Finally, \eqref{eq:goal_N} is trivial.

\section{Effective noise}\label{app:effective}

Relations of the form~\eqref{eq:condition_RN} with $f_4(\eta)=0$, \emph{i.e.} of the form
\begin{equation}
\mathrm{eff}(\cchan R_\eta)\subseteq \cchan N_{g_4(\eta)},
\end{equation}
are given in~\cite{bombin:2015:single-shot} within the proofs of the two theorems in that work, each corresponding to a different kind of stabilizer code family and a different kind of residual noise $\cchan N_\tau$.

The first case is that of self-correcting stabilizer codes and residual noise $\cchan N_\tau^\mathrm{exc}$, and the dependence of $\tau$ on $\eta$ is\begin{equation}
g_4(\eta)=\frac {(\eta/\eta_0)^{1/2}}{1-(\eta/\eta_0)^{1/2}},
\end{equation}
for some positive constant $\eta_0$ dependent on the code family.

The second case is that of a special kind of stabilizer subsystem codes (termed `$K$-confining', where $K$ is some positive constant for a given code family) and residual noise
\begin{equation}
\cchan N_\tau^\mathrm{loc}:= \synd{\cchan L_\tau},
\end{equation}
with a dependence
\begin{equation}
g_4(\eta)=\frac {(\eta/\eta_0)^{k}}{1-(\eta/\eta_0)^{k}},\qquad k:=\frac{1}{2(1+K)}.
\end{equation}
This result is not immediately applicable to the approach of section~\ref{sec:local_noise}, which requires $\cchan N_\tau=\cchan L_\tau$.
Fortunately, the proof given in~\cite{bombin:2015:single-shot} can be easily adapted.
Namely, in that proof for any given $\schan R\in \cchan R_\eta$ some $\schan D\in\cchan P$ is constructed with
\begin{equation}
\mathrm{eff} (R)= \synd{\schan D},
\end{equation}
and then it is proven that
\begin{equation}
\schan D\in\cchan L_\tau.
\end{equation}
In particular, the construction requires $\schan D=\stoc{p_i}{E_i}$ to be such that $\supp {E_i}$ is minimal for every $i$, with $\synd{E_i}$ fixed.
This constraint can be satisfied taking $\schan D'=\stoc{p_i}{\synd{E_i}}$ (as long as the correction operators $\omega_\sigma$ have minimal support:
Such a choice of correction operators is valid for these codes, as it yields a threshold for local noise~\cite{gottesman:2014:fault}).
If amended in this way, the proof shows that $\schan D'\in\cchan L_\tau$ and thus
\begin{equation}
\mathrm{eff} (R) \subseteq \cchan L_\tau.
\end{equation}

\bibliography{refs}

\end{document}